\documentclass[11pt]{article}
\usepackage{fullpage}
\usepackage{url}
\usepackage{xspace}

\usepackage{graphics}
\usepackage[dvips]{epsfig}

\usepackage{amsmath}
\usepackage{amssymb}
\usepackage{amsfonts}
\usepackage{graphicx}
\usepackage{algorithm}
\usepackage{comment}
\usepackage{algpseudocode}

\newtheorem{theorem}{Theorem}[section]

\newtheorem{claim}{Claim}[section]
\newtheorem{proposition}{Proposition}[section]

\newtheorem{fact}{Fact}[section]

\newcommand{\qed}{\hfill $\Box$ \bigbreak}
\newenvironment{proof}{\noindent {\bf Proof.}}{\qed}

\newcommand{\remove}[1]{}

\begin{document}

\baselineskip  0.2in %  0.2in %0.18in si on veut compact
\parskip     0.1in %    0.1in % 0.0in  pour compacter
\parindent   0.0in %    0.0in % 0.3in pour voir les paragraphes

\title{{\bf Deterministic Distributed Construction\\ of $T$-Dominating Sets in Time $T$}}
\date{}
\newcommand{\inst}[1]{$^{#1}$}

\author{
Avery Miller\inst{1},
Andrzej Pelc\inst{1}\\
\inst{1} Universit\'{e} du Qu\'{e}bec en Outaouais, Gatineau, Canada.\\
E-mails: \url{avery@averymiller.ca}, \url{pelc@uqo.ca}\\
}

\date{ }
\maketitle

\begin{abstract}
A $k$-dominating set is a set $D$ of nodes of a graph such that, for each node $v$, there exists a node $w \in D$ at distance at most $k$ from $v$.
Our aim is the deterministic distributed construction of small $T$-dominating sets in time $T$ in networks modeled as undirected $n$-node graphs and under the $\cal{LOCAL}$ communication model.

For any positive integer $T$, if $b$ is the size of a pairwise disjoint collection of balls of radii at least $T$ in a graph, then $b$ is an obvious lower bound on the size of a
$T$-dominating set. Our first result shows that, even on rings, it is impossible to construct a $T$-dominating set of size $s$ asymptotically $b$ (i.e., such that $s/b \rightarrow 1$) 
in time $T$.

In the range of time $T \in \Theta (\log^* n)$, the size of a $T$-dominating set turns out to be very sensitive to multiplicative constants in running time.
Indeed, it follows from \cite{KP}, that for time $T=\gamma \log^* n$ with  large constant $\gamma$, it is possible to construct a $T$-dominating set whose size is a small fraction of $n$. 
By contrast, we show that, for time $T=\alpha \log^* n $ for small constant $\alpha$, the size of a $T$-dominating set must be a large fraction of $n$. 

Finally, when $T \in o (\log^* n)$, the above lower bound implies that, for any constant $x<1$, it is impossible to construct
a $T$-dominating set of size smaller than $xn$, even on rings.  
On the positive side, we provide an algorithm that constructs a  $T$-dominating set of size $n- \Theta(T)$ on all graphs.

\vspace{1ex}

\noindent {\bf Keywords:} $T$-dominating set, LOCAL model, distributed algorithm, undirected networks
\end{abstract}

\vfill

\vfill

\thispagestyle{empty}
\pagebreak

%%%%%%%%%%%%%%%%%%%%%%%%%%%%%%%%%%%%%%%%%%%%%%%%%%%%%%%%%%%
\section{Introduction}
%%%%%%%%%%%%%%%%%%%%%%%%%%%%%%%%%%%%%%%%%%%%%%%%%%%%%%%%%%%

\subsection{Background}

A $k$-dominating set is a set $D$ of nodes of a graph with the property that for each node $v$ there exists a node $w \in D$ at distance at most $k$ from $v$.
Our aim is the deterministic distributed construction of small $T$-dominating sets in time $T$, in networks modeled as undirected graphs.
Such sets are important in many applications. For example, placing facilities (e.g., gas stations or restaurants in a town, or databases in a communication network) at nodes of a $T$-dominating set guarantees that every node will be at distance
at most $T$ from some facility. However, in order to take advantage of this proximity, every node should know a short path to some nearby facility.
Then a prospective customer will be able to reach a nearby gas station or restaurant from any street crossing, and a mobile agent situated at any node of a network will be able to reach a nearby database. 
{In many applications related to computer networks, in particular when the bandwidth is large (e.g., in optical networks) the time needed to send any message to a node at distance $r$ in the underlying graph is (proportional to) $r$. Hence we may assume that, given time $T$, each node can learn only the locations of databases situated at distance at most $T$ from it.}
%Hence we seek
%algorithms for constructing $T$-dominating sets in time $T$ which have the additional property of {\em representative awareness}: every node $v$ of the graph
%must learn a path of length at most $T$ to some node of the $T$-dominating set. 
This is the reason why, given some time $T$, we look for $T$-dominating sets
(and not, e.g., just for $(T+1)$-dominating sets).
If the constructed set is not $T$-dominating then there are nodes in the graph which do not become aware of any node in the chosen set within time $T$. For reasons of economy,
we want the constructed $T$-dominating set to be as small as possible.

\subsection{Model and Problem Description} 

The network is modeled as an undirected graph with $n$ labeled nodes.
Labels are drawn from the set of integers $\{1,\dots,L\}$, where $L$ is polynomial in $n$. Each node has a distinct label.
Initially each node knows its label, its degree, and parameters $L$, and $T$.

We use the extensively-studied $\cal{LOCAL}$ communication model \cite{Pe}. In this model, communication proceeds in synchronous
rounds and all nodes start simultaneously. In each round, each node
can exchange arbitrary messages with all of its neighbours and perform arbitrary local computations. 
Hence, the decisions of a node $v$ in round $r$ in any deterministic algorithm are a function of: (1) the subgraph induced by nodes at distance at most $r$ from $v$,
except for the edges between nodes at distance exactly $r$ from $v$; and (2)  the degrees of all nodes at distance $r$ from $v$.
The {\em time} of a task is the minimum number of rounds sufficient to complete it by all nodes. 

It is well known that the synchronous process of the $\cal{LOCAL}$  model can be simulated in an asynchronous network. This can be achieved 
by defining for each node separately its asynchronous round $i$;
in this round, a node performs local computations, then sends messages stamped $i$ to all neighbours, and  waits until it gets messages stamped $i$ from all neighbours.
To make this work, every node is required to send at least one (possibly empty) message with each stamp until termination.
Thus, all of our results can be translated for asynchronous networks by replacing ``time of completing a task''  by ``the maximum number of asynchronous rounds  to complete it, taken over all nodes''.

A deterministic algorithm working in time $T$ distributedly constructs a $T$-dominating set if, after $T$ rounds, some nodes output 1, all other nodes output 0,
and the nodes that output 1 form a $T$-dominating set. In all algorithms leading to upper bounds in this paper, every node additionally learns a path of length at most $T$ to some node of the $T$-dominating set. 

%The algorithm has {\em representative awareness} if every node that outputs 0 additionally outputs a path of length
%at most $T$ to some node that outputs 1. While every algorithm working in time $T$ that constructs a $k$-dominating set can be transformed into an algorithm with
%representative awareness working in time $T+k$ (nodes that output 1 communicate their position in the graph in additional $k$ rounds), if we insist on time $T$,
%it is not clear if the property of representative awareness is automatically satisfied for every algorithm constructing a $T$-dominating set. In all our algorithms we %explicitly guarantee this property.  

We use the following terminology. When $f(n) \in \Theta (g(n))$, we say that functions $f$ and $g$ have the same order of magnitude. When $f(n)/g(n)$ converges to 1,
we say that $f$ and $g$ are asymptotically equal.

\subsection{Our results} 
For a given time $T$, we give upper and lower bounds on the size of a $T$-dominating set that can be deterministically constructed in time $T$.
The main technical contribution of this paper are lower bounds that
are valid  even  on the class of rings.

For any positive integer $T$, if $b$ is the size of a pairwise disjoint collection of balls of radii at least $T$ in a graph, then $b$ is an obvious lower bound on the size of a
$T$-dominating set. Our first result shows that, even on rings, it is impossible to construct a $T$-dominating set of size $s$ asymptotically $b$ (i.e., such that $s/b \rightarrow 1$) 
in time $T$. Indeed, we prove that for rings (where there exist $T$-dominating sets of size $b=\lceil n/(2T+1) \rceil$) any $T$-dominating set constructed in time $T$ must be of size  larger than
$\lambda n/(2T+1)$, for any $\lambda <3/2$. By contrast, it follows from \cite{KP} that a $T$-dominating set of size $O(n\log^* n/(2T+1))$ can be constructed in time $T$
in any graph, which gives size $o(n)$ for time $\omega(\log^* n)$.

In the range of time $T \in \Theta (\log^* n)$, the size of a $T$-dominating set turns out to be very sensitive to multiplicative constants in running time.
Indeed,  it follows from \cite{KP} that, for time $T=\gamma \log^* n$ with  large constant $\gamma$, it is possible to construct a $T$-dominating set whose size is a small fraction of $n$. More precisely, the algorithm from \cite{KP} has the property that, for any constant $x>0$, there exists a positive constant $\gamma$ for which this algorithm produces, on any sufficiently large graph, a $T$-dominating set of size smaller than $xn$ in time $T=\gamma \log^* n $. 
By contrast, we show that, for time $T=\alpha \log^* n $ for small constant $\alpha$, the size of a $T$-dominating set must be a large fraction of $n$. More precisely, we prove
that, for any constant $x<1$, there exists a constant $\alpha>0$, such that any algorithm constructing a $T$-dominating set in time $T=\alpha \log^* n $ will produce
a set of size at least $xn$ on some ring of arbitrarily large size~$n$. 

Finally, moving to very short time, i.e., when $T \in o (\log^* n)$, the above lower bound implies that, for any constant $x<1$, it is impossible to construct
a $T$-dominating set of size smaller than $xn$, even on rings.  
On the positive side, we provide an algorithm that constructs a  $T$-dominating set of size $n- \Theta(T)$ on all graphs.

Thus our results show two gaps in the minimum size of a $T$-dominating set that can be constructed in time $T$: the first gap is while moving from time $\omega(\log^* n)$
to time $\Theta (\log^* n)$, when this size goes from $o(n)$ to $\Theta (n)$, and the second gap is while moving from time $\Theta (\log^* n)$ to time $o (\log^* n)$, 
when this size becomes larger than $xn$ for any constant $x<1$.

%--------------------------------------------------
\subsection{Related work}
%--------------------------------------------------
Distributed solutions of combinatorial optimization problems on graphs have been intensely studied in the last two decades. Research was aimed at fast
vertex coloring \cite{BE,BE1}, fast construction of maximal independent sets \cite{AABCHK,KMNR,MW}, of dominating and $k$-dominating sets \cite{KW,KP},
and of minimum weight spanning trees \cite{GKP, KP}. Various communication models have been used, ranging from the $\cal{LOCAL}$
model used in this paper, to the $\cal{CONGEST}$ model in which messages must be of logarithmic size \cite{KP}, the radio network model  \cite{MW}, 
and to the highly contrived beeping model
\cite{AABCHK} in which a node can transmit only one beep in each round.

In \cite{KM}, the authors studied how fast a capacitated minimum dominating set can be distributedly constructed. They showed that, for general graphs,
every distributed algorithm achieving a non-trivial approximation ratio (even for uniform capacities) must have a time complexity that essentially grows linearly with the network diameter. In \cite{JRS,KW}, randomized distributed solutions for dominating set approximation were presented. 
In \cite{LW} the authors prove that, for any $f(n)$-approximation of the minimum dominating set or maximum independent set on Unit Disk Graphs, the time $g(n)$
of finding this approximation must satisfy $f(n)g(n) \in \Omega(\log ^* n)$.  
The paper most closely related to the present work is \cite{KP}. The authors present a distributed algorithm to find a $k$-dominating set of size at most $n/(k+1)$ in arbitrary $n$-node graphs.
Their algorithm runs in time $O(k\log^* n)$ in the $\cal{CONGEST}$ model.

\section{A general lower bound}

The following useful fact is a straightforward consequence of the definition of a $T$-dominating set.

\begin{fact}\label{notsparse}
	For any $T$-dominating set $S$ in a ring $R$, there must be at least one member of $S$ in each segment of $2T+1$ nodes.
\end{fact}

Fact \ref{notsparse} implies that $\lceil \frac{n}{2T+1} \rceil$ is a lower bound on the size of a $T$-dominating set in rings. In fact, every ring of size $n$ has a $T$-dominating set
of this size. Our first result shows that in time $T \in o(n)$ we cannot construct a $T$-dominating set of size even asymptotic in this lower bound.
(For $T \in \Omega (n)$ this question is meaningless, since $ \frac{n}{2T+1}$ is then $O(1)$.)

\begin{theorem}
	Consider any constant $\lambda$ smaller than 3/2 and any algorithm $\mathcal{A}$  that runs in time $T \in o(n)$ and outputs a $T$-dominating set. For sufficiently large $n$, there exists a ring of size $n$ for which $\mathcal{A}$ outputs a $T$-dominating set of size greater than $\lambda\frac{n}{2T+1}$.
\end{theorem}
\begin{proof}
	For ease of exposition, assume that $2T+1$ divides $n$ and that $\frac{n}{2T+1}$ is divisible by 4. The proof can be modified if the latter assumptions are not satisfied.
	
	The high-level idea to prove our lower bound is as follows.
	We first execute algorithm $\mathcal{A}$  on a ring of size $n$ and we pick a representative member of the resulting $T$-dominating set
	in each of the segments of size $2T+1$ that form a partition of the ring. Next, we take representatives of even-numbered segments and partition them into consecutive pairs. For each pair, we construct a path consisting of their balls with radius $T$ and one additional node separating them. 
	We repeat this process on a different ring (disjoint from the first) to obtain additional paths of this type.
	We concatenate sufficiently many of these paths and add enough additional nodes to form a ring of size $n$. As a consequence, if we run algorithm $\mathcal{A}$ on this new ring, the representatives have the same balls with radius $T$
	as in the original ring, and hence act identically. Further, pairs of representatives are too far apart in the new ring to $T$-dominate all nodes, and consequently there must be an additional node in the $T$-dominating set between the representatives in each pair. This will imply our lower bound.  
	
	We now show how this idea is implemented. 
	Note that, if  $\lambda$ is a constant smaller than 3/2 and $T \in o(n)$, then $n - \frac{\lambda}{3}\frac{n}{2T+1}(4T+3) \geq 8T+4$.
	Let $R_1$ be the ring obtained from the path of nodes $[1,\ldots,n]$ by adding the edge $\{1,n\}$, and let $R_2$
	be the ring obtained from the path of nodes $[n+1,\ldots,2n]$ by adding the edge $\{n+1,2n\}$. 
	We illustrate our construction using ring $R_1$ to obtain paths $P_0,\dots, P_r$, and an analogous construction using ring $R_2$ will produce paths  
	$P_{r+1},\dots, P_{2r+1}$.
	
	We partition this ring into segments of size $2T+1$, namely, for each integer $i \in \{0,\ldots,\frac{n}{2T+1} - 1\}$, denote by $S_i$ the segment $[(2T+1)i+1,\ldots,(2T+1)(i+1)]$. Execute algorithm $\mathcal{A}$ on $R$, and call this execution $\mathcal{A}_1$. For each $i  \in \{0,\ldots, \frac{n}{2T+1} - 1\}$, let $m_i$ be the node in $S_i$ with smallest label that outputs 1 in this execution. By Fact \ref{notsparse}, $m_i$ is well-defined.
	
	Next, for each integer $j  \in \{0,\ldots, \frac{1}{2}\left(\frac{n}{2T+1}\right) - 1\}$, define $H_{2j}$ to be the path consisting of $2T+1$ nodes centered at $m_{2j}$, namely $[m_{2j}-T,\ldots,m_{2j},\ldots,m_{2j}+T]$. Note that, for every $j < j'$, the paths $H_{2j}$ and $H_{2j'}$ are disjoint since there are at least $2T+1$ nodes between $m_{2j}$ and $m_{2j'}$ in $R$ (for example, the nodes in $S_{2j+1}$.)

	Next, for each integer $k \in \{0,\ldots, r\}$, where $r=\frac{1}{4}\left(\frac{n}{2T+1}\right) - 1$, define $P_k$ to be the path obtained by taking node $v_k = (2T+1)(4k+1)+T+1$ (i.e., the middle node of segment $S_{4k+1}$), the paths $H_{4k},H_{4k+2}$, and adding the edges $\{v_k,m_{4k}+T\}$ and $\{v_k,m_{4k+2}-T\}$. Note that, in $P_k$, there exist $2T+1$ nodes between $m_{4k}$ and $m_{4k+2}$. Further, $|P_k| = 4T+3$.
	
	This concludes the construction of paths $P_0,\dots, P_r$. In the construction of paths $P_{r+1},\dots, P_{2r+1}$, the range for index $i$
	starts at $\frac{n}{2T+1}$ , the range of index $j$ starts at $\frac{1}{2}\frac{n}{2T+1}$,  and the range of index $k$ starts at $\frac{1}{4}\frac{n}{2T+1}$.

	Finally, let $c  = \lfloor\frac{\lambda}{3}\frac{n}{2T+1}\rfloor$, and construct a ring $R_c$ as follows.
	Choose $c$ paths $G_0 \cup \ldots \cup G_{c-1}$ from the set of paths $\{P_0,\dots,  P_{2r+1}\}$.   
	Let $G_c$ be a path consisting of $n-c(4T+3)$ nodes whose labels do not appear in $G_0 \cup \ldots \cup G_{c-1}$. Construct $R_c$ by concatenating the paths $G_0,\ldots,G_{c}$ and adding an edge between the endpoints of this path. Note that $|R_c| = n$. The above construction is illustrated in Figure \ref{construction}.
	
	\begin{figure}[!ht]
		
		\begin{center}
			\includegraphics[scale=0.6]{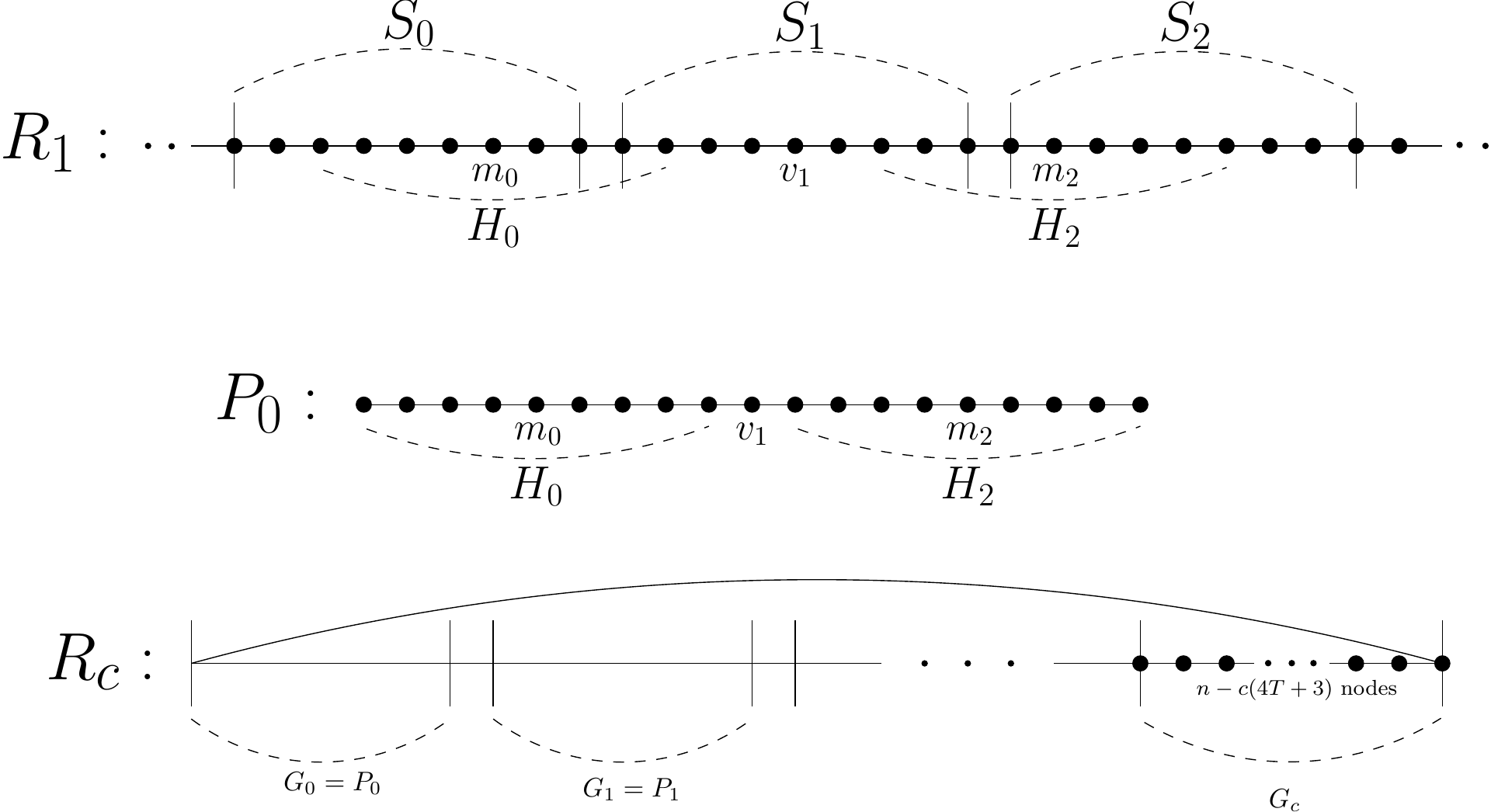}
		\end{center}
		\caption{An example of the construction of $R_c$ when $T=4$.}
		\label{construction}
	\end{figure}

	We now execute algorithm $\mathcal{A}$ on $R_c$, and call this execution $\mathcal{A}_2$. First, note that, for each $k \in \{0,\ldots,c-1\}$, the nodes $m_{4k}$ and $m_{4k+2}$ cannot distinguish between executions $\mathcal{A}_1$ and $\mathcal{A}_2$, so they will both output 1 in $\mathcal{A}_2$. Further, for each $k \in \{0,\ldots,c-1\}$, since there are $2T+1$ nodes between $m_{4k}$ and $m_{4k+2}$, Fact \ref{notsparse} implies that there must be at least one node between $m_{4k}$ and $m_{4k+2}$ that outputs 1. Hence, for each $k \in \{0,\ldots,c-1\}$, there are 3 nodes in $G_k$ that output 1. Finally, recall that $|G_c| =  n-c(4T+3) \geq n - \frac{\lambda}{3}\frac{n}{2T+1}(4T+3)$, and, by our choice of $\lambda$, this is bounded below by $8T+4$. By Fact \ref{notsparse}, there are at least 4 nodes in $G_c$ that output 1. Therefore, the number of nodes that output 1 in execution $\mathcal{A}_2$ is bounded below by $3c+4=3\lfloor\frac{\lambda}{3}\frac{n}{2T+1}\rfloor +4 \geq 3\frac{\lambda}{3}\frac{n}{2T+1} - 3 + 4 > \lambda\frac{n}{2T+1}$.
\end{proof}

We do not know if  a $T$-dominating set of size $\Theta(\frac{n}{2T+1})$ can be constructed in time $T$, even on rings. The best known upper bound of $O(\frac{n\log^ * n}{2T+1})$ on the size of a $T$-dominating set that can be constructed in any graph in time $T$
follows from \cite{KP}.

\section{Time $\Theta(\log^*n)$}

In the range of time $T \in \Theta (\log^* n)$, the size of a $T$-dominating set turns out to depend on the multiplicative constant in running time.
First notice that, for time $T=\gamma \log^* n$ with large constant $\gamma$, {Theorem 4.4 of} \cite{KP} implies that a $T$-dominating set of size  $xn$ for a small constant $x$ can be constructed in every graph. More precisely we have the following proposition.

\begin{proposition}\label{ub*}
	For every positive constant $x < 1$, there exists a positive constant $\gamma$ such that, for all sufficiently large networks of size $n$ and when $T=\gamma\log^*{n}$, there is an algorithm producing a $T$-dominating set of size at most $xn$ in time $T$ .
\end{proposition}

In contrast with the above positive result, we now show that, for small time $T \in \Theta(\log ^* n)$, the size of any $T$-dominating set produced in time $T$ must be a large fraction of $n$, even on rings with node labels from $\{1,\ldots,n\}$.

\begin{theorem}\label{lb}
	For any positive constant $x < 1$, there exists a positive constant $\alpha$ such that, for any algorithm $\cal A$ that takes input $T = \lfloor\alpha\log^*{n}\rfloor$ and finds a $T$-dominating set in time $T$ on all rings of size $n$, algorithm $\cal A$  produces a $T$-dominating set of size greater than $xn$, for arbitrarily large $n$.
\end{theorem}
\begin{proof}
	Choose constants $\beta$ and $n_1$ such that, for any $n \geq n_1$, any algorithm working in time at most $\beta\log^*{n}$ on rings of size $n$ fails to produce a proper 8-colouring on some
	ring of size $n$. From \cite{Lin}, we know that the time $t$ needed to 8-colour rings of size $n$ satisfies $\log^{2t}(n) \leq 8$. It follows that $2t \geq \log^*{n} - 2$, so $t \geq \frac{1}{2}\log^*{n} - 1$. When $\log^*{n} > 3$, it follows that $\frac{1}{2}\log^*{n} - 1 \geq \frac{1}{2}\log^*{n} - \frac{1}{3}\log^*{n} = \frac{2}{3}\log^*{n}$. Therefore, $\beta = \frac{2}{3}$ and $n_1 = 16$ are suitable choices.
	
	Choose any integer constant $y$ such that $\frac{y-2}{y} > x$. Note that $x > 0$ implies $y > 2$. Let $\alpha = \frac{\beta}{4y}$. To obtain a contradiction assume that, for some positive integer $n_2$, we have an algorithm $\mathcal{A}$ with input $T$ that, for all $n \geq n_2$, when $T = \lfloor\alpha\log^*{n}\rfloor$, the algorithm finds a $T$-dominating set of size at most $xn$ on every ring of size $n$.
	To prove the theorem we show that,
	for any integer $n_0 > 0$, there exists an $n \geq \max\{n_0,n_1\}$ such that, for every ring $R$ of size $n$ with node labels from $\{1,\ldots,n\}$, we can 8-colour $R$ in time less than $\beta\log^*{n}$, which contradicts the choice of $\beta$ and $n_1$.
	
	Choose $n$ such that $\min\{\lfloor y\alpha\log^*{n} \rfloor, n\} \geq n_2$ and $n \geq \max\{n_0,n_1\}$. Consider any ring $R$ of size $n$ with node labels from $\{1,\ldots,n\}$. At a high level, our 8-colouring of $R$ will work as follows. Execute algorithm $\mathcal{A}$ with input $T = \lfloor\alpha\log^*{n}\rfloor$ on $R$ and call this execution $\mathcal{A}_1$. By assumption, when execution $\mathcal{A}_1$ terminates, a $T$-dominating set in $R$ of size at most $xn$ has been constructed. The nodes in this set will be called \emph{members}, and all other nodes will be called \emph{non-members}. A maximal segment of $R$ that contains only members (respectively, non-members) will be called a \emph{stretch} of members (respectively, non-members.) We would like to distributedly constant-colour each stretch, which can be done if all nodes see the
	boundaries of the stretch to which they belong. Hence, in order to break up long stretches,
	the members that do not have nearby non-members in both directions will execute algorithm $\mathcal{A}$ again, but with a
	carefully chosen input $T'$ smaller than $\lfloor\alpha\log^*{n}\rfloor$. Call this execution $\mathcal{A}_2$. The nodes in the $T'$-dominating set constructed in this execution will be called \emph{survivors}, and all other nodes that were involved in execution $\mathcal{A}_2$ will be called \emph{non-survivors}. As before, a maximal segment of $R$ that contains only survivors (respectively, non-survivors) will be called a \emph{stretch} of survivors (respectively, non-survivors.) We will prove that stretches of 
	members, non-members, survivors and non-survivors are now short enough for their elements to see their boundaries. 
	Finally, the nodes will properly two-colour the stretch to which they belong: stretches of non-members  using colours $\{1,2\}$, stretches of members using colours $\{3,4\}$, stretches of non-survivors using colours $\{5,6\}$, and stretches of survivors using colours $\{7,8\}$. This will properly colour the entire ring. For a full description of the algorithm executed at each node $v$, see Algorithm \ref{colouralg}. We denote by $\ell(v)$ the label of node $v$. 
	
	%We now show how to 8-colour $R$ using $\frac{\beta}{2}\log^*{n}$ rounds. Each node $x$ executes

	\begin{algorithm}[H]
		\footnotesize
		\caption{\texttt{EightColourRing}$(n,x)$}
		\begin{algorithmic}[1]
			\State Let $y$ be an integer such that $\frac{y-2}{y} > x$
			\State $\alpha \leftarrow \frac{\beta}{4y}$
			\State $T \leftarrow \lfloor \alpha\log^*{n} \rfloor$
			\State Using $2yT$ rounds, learn the labels of all nodes within distance $2yT$ of $v$
			\State For each node $w$ within distance $2yT - T$ of $v$:
			\State\indent  Run algorithm $\mathcal{A}$ with input $T$ on $w$. Call this execution $\mathcal{A}_1(w)$\label{executeA1}
			\State $\mathcal{M} \leftarrow$ set of nodes $w$ such that $\mathcal{A}_1(w)$ outputs 1 /* set of members */ 
			\State $\mathcal{NM} \leftarrow$ set of nodes $w$ such that $\mathcal{A}_1(w)$ outputs 0  /* set of non-members */ 
			\State If $v \in \mathcal{NM}$:
			
			\State \indent $(v_1,\ldots,v_k) \leftarrow$ the stretch of nodes in $\mathcal{NM}$ containing $v$, with $\ell(v_1) < \ell(v_k)$\label{NMFindStretch}
			\State\indent  $v$ gets colour 1 if its distance from $v_1$ is even
			\State\indent  $v$ gets colour 2 if its distance from $v_1$ is odd
			
			\State If $v \in \mathcal{M}$:

			\State\indent  If  $v$ belongs to a stretch of length at most $yT$: \label{short}
			\State\indent\indent  $(v_1,\ldots,v_k) \leftarrow$ the stretch of nodes in $\mathcal{M}$ containing $v$, with $\ell(v_1) < \ell(v_k)$\label{MFindStretch}
			
			\State\indent\indent   $v$ gets colour 3 if its distance from $v_1$ is even
			\State\indent \indent  $v$ gets colour 4 if its distance from $v_1$ is odd
			
			\State\indent  Else:
			
			\State\indent \indent  For each node $w$ in the same stretch as $v$ and within distance $yT$ from $v$:
			\State\indent\indent\indent Run algorithm $\mathcal{A}$ with input $T' = \lfloor \alpha\log^*{yT} \rfloor$ on $w$. Call this execution $\mathcal{A}_2(w)$\label{executeA2}
			\State\indent \indent  $\mathcal{S} \leftarrow$ set of nodes $w$ where $\mathcal{A}_2(w)$ outputs 1 /* set of survivors */ 
			\State\indent \indent  $\mathcal{NS} \leftarrow$ set of nodes $w$ where $\mathcal{A}_2(w)$ outputs 0 /* set of non-survivors */ 
			
			\State\indent \indent  If $v \in \mathcal{NS}$:
			
			\State\indent \indent \indent  $(z_1,\ldots,z_m) \leftarrow$ stretch of nodes in $\mathcal{NS}$ containing $v$, with $\ell(z_1) < \ell(z_m)$\label{NSFindStretch}
			\State\indent \indent \indent  $v$ gets colour 5 if its distance from $z_1$ is even
			\State\indent \indent \indent  $v$ gets colour 6 if its distance from $z_1$ is odd
			
			\State\indent \indent  If $v \in \mathcal{S}$:
			
			\State\indent \indent \indent  $(z_1,\ldots,z_m) \leftarrow$ stretch of nodes in $\mathcal{S}$ containing $v$, with $\ell(z_1) < \ell(z_m)$\label{SFindStretch}
			\State\indent \indent \indent  $v$ gets colour 7 if its distance from $z_1$ is even
			\State\indent \indent \indent  $v$ gets colour 8 if its distance from $z_1$ is odd

			%\end{enumerate}
		\end{algorithmic}
		\label{colouralg}
	\end{algorithm}

	%First, note that the algorithm uses $2yT = 2y\alpha\log^*{n} = 2y\frac{\beta}{4y}\log^*{n} = \frac{\beta}{2}\log^*{n}$ communication rounds, as desired. Next, since $\mathcal{A}$ uses $T$ rounds, note that it is possible to perform line 2, since $x$ has learned the $T$-neighbourhood of all nodes within distance $2yT-T$. It remains to show that line 3 is possible, that is, that after locally simulating algorithm $\mathcal{A}$ for all nodes within distance $2yT-T$, node $x$ can determine the stretch $(x_1,\ldots,x_k)$ to which it belongs. Let $x_0$ and $x_{k+1}$ be the neighbours of $x_1$ and $x_k$, respectively, that are not contained in $\{x_1,\ldots,x_k\}$. The distance from $x$ to each node in $\{x_0,\ldots,x_{k+1}\}$ is at most $k$. By Claims \ref{nonmembers} and \ref{members} and the fact that $y>2$, it follows that $k \leq yT$. So, since $y > 2$, the distance from $x$ to each node in $\{x_0,\ldots,x_{k+1}\}$ is at most $yT \leq 2yT - T$. Namely, after performing line 2, $x$ has determined whether or not each node in $\{x_0,\ldots,x_{k+1}\}$ belongs to the dominating set constructed by $\mathcal{A}$, from which it concludes that it belongs to stretch $(x_1,\ldots,x_k)$.
	
	% First, we show that $yT > n_0$. (To do: show this. Depending on actual value of $\beta$, we'll need to re-define $n_0$ appropriately by multiplying by some constant.)

	We now prove some useful facts about the number of nodes in any stretch of $R$ after executing \texttt{EightColourRing}. First, notice that any stretch of 
	members that did not execute algorithm $\cal A$ on line \ref{executeA2} is of length at most $yT$ (see line \ref{short}).  
	The next two claims follow from the fact that, since $\mathcal{A}$ produces a $T$-dominating set, every non-member (resp., non-survivor) is at distance at most $T$ from at least one member (resp., survivor).
	
	\begin{claim}\label{nonmembers}
		Every stretch of non-members contains at most $2T$ nodes.
	\end{claim}
	\begin{claim}\label{nonsurvivors}
		Every stretch of non-survivors contains at most $2T$ nodes.
	\end{claim}
	
	The following claim implies that long stretches of members are broken up into short stretches of survivors or of non-survivors.
	
	\begin{claim}\label{survivors}
		Every stretch of survivors contains less than $yT$ nodes.
	\end{claim}
	We prove the claim by way of contradiction. Assume that there exists a stretch of survivors containing a sequence of nodes $(z_1,\ldots,z_{yT})$. By definition, this means that the nodes $z_1,\ldots,z_{yT}$ outputted 1 in the execution $\mathcal{A}_2$, i.e.,  when provided $T' = \lfloor \alpha\log^*{yT} \rfloor$ as input. So, we construct a ring $R'$ of size $yT$ such that
	when algorithm $\cal A$ is executed on $R'$ with input $T' = \lfloor \alpha\log^*{yT} \rfloor$, nodes $z_{T+1},\ldots,z_{(y-1)T}$ output the same value as in execution $\mathcal{A}_2$ on $R$. In particular, we obtain $R'$ from $R$ by taking the segment $(z_1,\ldots,z_{yT})$ and adding the edge $\{z_1,z_{yT}\}$. We now consider the execution of $\mathcal{A}$ on $R'$ with input $T' = \lfloor \alpha\log^*{yT} \rfloor$, which we will call $\mathcal{A}_3$. Execution $\mathcal{A}_3$ consists of $\lfloor\alpha\log^*{(yT)}\rfloor = \lfloor\alpha\log^*{(y\lfloor\alpha\log^*{n}\rfloor)}\rfloor \leq \lfloor\alpha\log^*{n}\rfloor = T$ rounds. In particular, this means that the executions $\mathcal{A}_2$ and $\mathcal{A}_3$ are indistinguishable to each of the nodes $z_{T+1},\ldots,z_{(y-1)T}$ since they have the same balls with radius $T$ (and, therefore, the same balls with radius $T'$) in both executions. It follows that $z_{T+1},\ldots,z_{(y-1)T}$ output 1 in execution $\mathcal{A}_3$, hence they belong to the $T'$-dominating set constructed in execution $\mathcal{A}_3$.  Consequently, the $T'$-dominating set constructed in execution $\mathcal{A}_3$ has size at least $(y-2)T = \frac{(y-2)T}{yT}yT = \frac{y-2}{y}|R'| > x|R'|$. Finally, by our choice of $n$, it follows that $|R'| = yT = y\lfloor\alpha\log^*{n}\rfloor \geq n_2$. But this contradicts the assumption that, for every ring $R'$ of size at least $n_2$, algorithm $\mathcal{A}$, when given input $T'= \lfloor\alpha\log^*{|R'|}\rfloor$, produces a $T'$-dominating set of size at most $x|R'|$. This concludes the proof of the claim.
	
	% (to do)
	% Else, create ring using sequence of $yT$ nodes, adding 
	
	Next, we show that the colouring can be carried out. In particular, at lines \ref{NMFindStretch}, \ref{MFindStretch}, \ref{NSFindStretch}, and \ref{SFindStretch} of \texttt{EightColourRing}, node $v$ must identify all of the nodes in $R$ that belong to the stretch containing $v$. The following two claims show that this is possible.
	
	\begin{claim}
		At lines \ref{NMFindStretch} and \ref{MFindStretch}, the sequence $(v_1,\ldots,v_k)$ can be determined by $v$.
	\end{claim}
	
	In order to prove the claim, 
	let $v_0$ and $v_{k+1}$ be the neighbours of $v_1$ and $v_k$, respectively, that are not contained in $\{v_1,\ldots,v_k\}$. The distance from $v$ to each node in $\{v_0,\ldots,v_{k+1}\}$ is at most $k$. 
	
	First, consider line \ref{NMFindStretch}. By Claim \ref{nonmembers} and since $y>2$, the distance from $v$ to each node in $\{v_0,\ldots,v_{k+1}\}$ is at most $2T \leq 2yT - T$. Consequently, after performing line \ref{executeA1}, $v$ has determined that nodes $v_0$ and $v_{k+1}$ have output 1 and nodes $v_i$ for $0<i<k+1$ have output 0, from which it deduces that it belongs to stretch $(v_1,\ldots,v_k)$.
	
	At line \ref{MFindStretch}, since $y > 2$ and the condition in line \ref{short} evaluated to true, the distance from $v$ to each node in $\{v_0,\ldots,v_{k+1}\}$ is at most $yT \leq 2yT - T$. Consequently, after performing line \ref{executeA1}, $v$ has determined that nodes $v_0$ and $v_{k+1}$ have output 0 and nodes $v_i$ for $0<i<k+1$ have output 1, from which it deduces that it belongs to stretch $(v_1,\ldots,v_k)$. This concludes the proof of the claim.

	\begin{claim}
		At lines \ref{NSFindStretch} and \ref{SFindStretch}, the sequence $(z_1,\ldots,z_m)$ can be determined by $v$.
	\end{claim}
	In order to prove the claim, 
	let $z_0$ and $z_{m+1}$ be the neighbours of $z_1$ and $z_m$, respectively, that are not contained in $\{z_1,\ldots,z_m\}$.
	
	At line \ref{NSFindStretch}, by Claim \ref{nonsurvivors} and since $y>2$, the distance from $v$ to each node in $\{z_0,\ldots,z_{m+1}\}$ is at most $2T \leq yT$.
	Consequently, after performing line \ref{executeA2}, $v$ has determined that nodes $z_0$ and $z_{m+1}$ have output 1 and nodes $z_i$ for $0<i<m+1$ have output 0, from which it deduces that it belongs to stretch $(z_1,\ldots,z_m)$.

	Next, consider line \ref{SFindStretch}. By Claim \ref{survivors}, the distance from $v$ to each node in $\{z_0,\ldots,z_{m+1}\}$ is at most $yT$.
	Consequently, after performing line \ref{executeA2}, $v$ has determined that nodes $z_0$ and $z_{m+1}$ have output 0 and nodes $z_i$ for $0<i<m+1$ have output 1, from which it deduces that it belongs to stretch $(z_1,\ldots,z_m)$. 
	This concludes the proof of the claim.
	
	Finally, since \texttt{EightColourRing} properly 2-colours each stretch, and every two neighbouring stretches use disjoint sets of colours, it follows that $R$ has been properly 8-coloured. The number of communication rounds used by \texttt{EightColourRing} is $2yT \leq 2y\frac{\beta}{4y}\log^*{n}< \beta\log^*{n}$.
\end{proof}

The above theorem implies that
there exists a positive constant $\alpha$ such that, for any algorithm $\cal A$ that takes input $T = \lfloor\alpha\log^*{n}\rfloor$ and finds a $T$-dominating set in time $T$ on all rings of size $n$, algorithm $\cal A$  produces a $T$-dominating set of size $\Omega(n)$, for arbitrarily large $n$. It is interesting to compare this result with the
lower bound from \cite{LW}. In particular, when restricting attention to constructing $T$-dominating sets in time $T=\alpha \log^*n$, the result from \cite{LW} only implies 
the lower bound $\Omega(n/\log^*n)$ on the size of the $T$-dominating set, regardless of the choice of the constant $\alpha$.

\section{Time $o(\log^*n)$}

We first observe that
Theorem \ref{lb} implies a strict lower bound on the size of $T$-dominating sets that can be constructed in time $T \in o(\log^*n)$. Indeed, in this case,
for any constant $\alpha >0$ we have
$T \leq \alpha \log^* n$ for sufficiently large $n$. Hence Theorem \ref{lb} implies the following result.

\begin{proposition}
	Suppose that  $T \in o(\log^*n)$. For any positive constant $x < 1$,  any algorithm that takes input $T$ and finds a $T$-dominating set in time $T$ on all rings of size $n$ produces a $T$-dominating set of size greater than $xn$, for arbitrarily large $n$.
\end{proposition}

On the positive side, we show that, while a trivial $T$-dominating set consists of all $n$ nodes, we can reduce this size in time $T$ by a number of nodes proportional to $T$. The following algorithm is executed by a node with label $\ell$.

\begin{algorithm}[H]
	\caption{\texttt{ChooseSmallest}$(T)$}
	\footnotesize
	\begin{algorithmic}[1]
		\State $output \leftarrow 0$
		\State $r \leftarrow \lfloor T/2 \rfloor$
		
		\State Using $r$ communication rounds, get the set $S$ of labels of all nodes within distance $r$\label{learn}
		\State $min \leftarrow $ smallest label in $S$
		\State Using $r$ communication rounds, get the set $M$ of values of $min$ at all nodes within distance~$r$\label{getmins}
		\State {\bf If} $\ell \in M$:
		\State \indent $output \leftarrow 1$\label{join}
		\State {\bf return} $output$
		
	\end{algorithmic}
	\label{choosesmallest}
\end{algorithm}

\begin{theorem}
	For any positive integer $T$, Algorithm {\tt ChooseSmallest} produces a $T$-dominating set of size at most $n-\lfloor T/2 \rfloor$ in time $T$.  
\end{theorem}

\begin{proof}
	First, observe that Algorithm \ref{choosesmallest} uses $2\lfloor T/2 \rfloor \leq T$ communication rounds. 
	
	Next, we show that the nodes that output 1 form a $\lfloor T/2 \rfloor$-dominating set. To see why, consider an arbitrary node $v$, and note that the value of $min$ at $v$ specifies the label of some node $w$ within distance $\lfloor T/2 \rfloor$ from $v$. At line \ref{getmins} in node $w$'s execution of Algorithm \ref{choosesmallest}, $w$ will learn the value of $min$ at $v$, so $w$ will change its output to 1 at line \ref{join}. Thus, $v$ is dominated by $w$.
	
	Finally, we show that at most $n-\lfloor T/2 \rfloor$ nodes output 1. It is sufficient to show that all of the nodes with the largest $\lfloor T/2 \rfloor$ labels output 0. Consider any node $w$ such that fewer than $\lfloor T/2 \rfloor$ nodes in the network have a label larger than $w$'s label. To obtain a contradiction, assume that $w$ outputs 1. It follows from the algorithm's description that there is some node $z$ (possibly equal to $w$) within distance $\lfloor T/2 \rfloor$ from $w$ such that the value of $min$ at $z$ is equal to $w$'s label. Therefore, $w$'s label is the smallest out of all of the labels of nodes within distance $\lfloor T/2 \rfloor$ from $z$. However, there are at least $\lfloor T/2 \rfloor + 1$ nodes within distance $\lfloor T/2 \rfloor$ from $w$, which means that at least $\lfloor T/2 \rfloor$ of them have a label larger than $w$'s label. This contradicts the assumption that fewer than $\lfloor T/2 \rfloor$ nodes in the network have a label larger than $w$'s label. It follows that all of the nodes with the largest $\lfloor T/2 \rfloor$ labels output 0, which gives the desired upper bound on the number of nodes that output 1.
\end{proof}

Note that, on line \ref{learn} of Algorithm {\tt ChooseSmallest}$(T)$, each node $v$ learns a short path to the node with label $min$, i.e., to some node in the $T$-dominating set.

\section{Conclusion}

We established upper and lower bounds on the size of a $T$-dominating set that can be constructed in time $T$, for various times $T$.
While the remaining gaps between these bounds are not large, several interesting problems remain open.

In the time range $T \in \omega(\log^* n)$, it remains open if our lower bound $\lambda n/(2T+1)$, for any $\lambda <3/2$, can be sharpened to $\Omega(n\log^* n/(2T+1))$,
i.e., if the upper bound following from \cite{KP} has optimal order of magnitude.

When $T \in \Theta(\log^* n)$, probably the most interesting question concerns the upper end of this time range. Is there a constant $C$ so large that a $T$-dominating set of size
$o(n)$ can be constructed in time $T=C \log^*n$? (The result from \cite{KP} implies that we can construct such a set whose size is an arbitrarily small constant fraction of $n$.) 
Another question concerns determining the minimum time $T$ to construct a $T$-dominating set whose size is a given fraction of $n$. More precisely,  for a given constant $0<x<1$, what is the minimum constant $\xi$ such that a $T$-dominating set of size $xn$ can be constructed in time
$T=\xi \log^* n$? %(Our results give an interval to which such a constant $\xi$ must belong, for any $x$.) 

In the time range $T \in o(\log^* n)$, our results leave very little room, as size $xn$ of a $T$-dominating set, for any constant $x<1$, is excluded.
Nevertheless, it remains open if our lower bound can be sharpened to $n-\Theta(T)$ for such small values of $T$.

{In this paper we chose the $\cal{LOCAL}$ model, in which nodes can send messages of arbitrary size in each round.
	This is a reasonable assumption when the bandwidth is large, e.g., in optical networks. When the size of the bandwidth
	is more restricted, it would be more suitable to use the $\cal{CONGEST}$ model, in which only messages of size logarithmic in the size of the network can be sent in each round. It remains open how our results change in such a model. Of course, the lower bounds valid for the $\cal{LOCAL}$ model still hold for the more restrictive $\cal{CONGEST}$ model, but, for example, the time of Algorithm 2 would change, as it calls for sending large messages, which could potentially use many rounds just to reach immediate neighbours.}

\section*{Acknowledgments}
This research was partially supported by NSERC discovery grant 8136 -- 2013 and by the Research Chair in Distributed Computing at the Universit\'e du Qu\'{e}bec en Outaouais.

%%%%%%%%%%%%%%%%%%%%%%%%%%%%%%%%%%%%%%%%%%%%%%%%%%%%%%%%%%%
\bibliographystyle{elsarticle-num}

%%%%%%%%%%%%%%%%%%%%%%%%%%%%%%%%%%%%%%%%%%%%%%%%%%%%%%%%%%% 

\end{document}